\documentclass{article}
\usepackage{geometry}
\usepackage{graphicx}
\usepackage{dcolumn}
\usepackage{subcaption}
\usepackage{mathdots}
\usepackage{amsmath,amssymb,amsthm}
\usepackage{physics}
\usepackage{comment}

\usepackage{bm}

\theoremstyle{plain}
\newtheorem{theorem}{Theorem}[section]
\newtheorem{lemma}[theorem]{Lemma}
\newtheorem{corollary}[theorem]{Corollary}
\newtheorem{definition}[theorem]{Definition}

\newcommand{\ii}{\mathrm{i}}
\newcommand{\e}{\mathrm{e}}

\begin{document}

\AtEndDocument{%
  \par
  \medskip
  \begin{tabular}{@{}l@{}}%
    \textsc{Bruno Chagas}\\
    \textsc{Irish Centre for High-End Computing in Dublin}\\ \textsc{National University of Ireland, Galway, Ireland}\\
    \textit{E-mail address}: \texttt{bruno.chagas@ichec.ie} \\ \ \\
    \textsc{Rodrigo Chaves}\\
    \textsc{Department of Computer Science,} \\ \textsc{Universidade Federal de Minas Gerais, Belo Horizonte, Brazil}\\
    \textit{E-mail address}: \texttt{rodchaves@ufmg.br} \\ \ \\
    \textsc{Gabriel Coutinho}\\
    \textsc{Department of Computer Science,} \\ \textsc{Universidade Federal de Minas Gerais, Belo Horizonte, Brazil}\\
    \textit{E-mail address}: \texttt{gabriel@dcc.ufmg.br}
  \end{tabular}}

\title{Why and how to add direction to a quantum walk}

\author{Rodrigo Chaves, Bruno Oliveira Chagas, Gabriel Coutinho}

\date{\today}

\maketitle

\begin{abstract}
We formalize the treatment of directed (or chiral) quantum walks using Hermitian adjacency matrices, bridging two developing fields of research in quantum information and spectral graph theory. We display results and simulations which highlight the conceptual differences between having directions encoded in the Hamiltonians or not. This leads to a construction of a new type of quantum phenomenon: zero transfer between pairs of sites in a connected coupled network, which is only possible in the directed model we study. Our main result is a description of several families of directed cycles that admit zero transfer.
\end{abstract}

\section{\label{sec:intro}Introduction}

The interplay between graph theory and quantum information has been widely investigated \cite{Aharonov2002, Childs2002, Aharonov1993, Cabello2014, Amaral2018}. In the context of quantum walks, this is caused by the very natural way in which a graph models the quantum system defined by Hamiltonians which are symmetric and real and encode couplings between qubits \cite{Christandl2004,Christandl2005}. The amount of works devoted to understand this connection is proportional to the extensive existing literature of spectral graph theory, and most of which is devoted to study of how the spectral properties of the adjacency matrix or versions of the Laplacian matrix for undirected graphs relate to quantum properties \cite{Godsil2011,Godsil2012,Alvir2016}.

Recently there has been a surged interest in using Hermitian complex matrices, which very conveniently provide models for arc direction in directed graphs \cite{Mohar2016}. From the graph theoretic point of view, several works have focused on studying how the combinatorics of arc direction is manifested in spectral properties of Hermitian adjacency matrices \cite{Mohar2017}. This motivates the seek of which quantum information phenomenon can be achieved in this more general setting of Hermitian quantum walks which in contrast are not available for the restricted undirected model \cite{Lato2020,Cameron2014}. 

Notably, the possibility of setting up a system in which a qubit state hops from one site to any other at different times with perfect fidelity and with no time-dependent control on the Hamiltonian can only be achieved with Hermitian complex Hamiltonians \cite{Connelly2017}. 

Achieving state transfer over large distances in finite qubit networks has also been a task pursued since early works \cite{Bose2003,Bose2007}. This has been found possible with real symmetric Hamiltonians only upon modulating the couplings with high energy \cite{KayReviewPST}, and unfortunately allowing for complex weightings is not helpful, as we verify in Section \ref{sec:ctqw2}.

Investigations of directed quantum walks on chains (paths) and rings (cycles) have appeared recently \cite{Zimboras2013, Sett2019}; special attention has been paid to the possibility of shielding part of the network to the state evolution initiated somewhere else. This is our main result: we show in Section \ref{sec:ztrans} how to achieve zero transfer at all times between antipodal pairs of sites in ring networks upon adding certain complex weights. This generalizes known examples in the literature, and displays an interesting connection between some old theorems, opening avenues for future exploration.

The paper is organized as follows. We begin our treatment formalizing our Hamiltonian model in Section \ref{sec:ctqw2}, showing how it decomposes into invariant subspaces corresponding to the $k$-excitation subspaces, and how the block corresponding to the $1$-excitation subspace is precisely the Hermitian adjacency matrix that appears in Christandl \textit{et al.}\cite{Mohar2017}. We also add simulations with interesting results for the complete graphs (all qubits coupled). Then, we introduce and show the result mentioned above: zero transfer in cycles in Section \ref{sec:ztrans}.

\section{\label{sec:ctqw2} Hamiltonian for directed quantum walks}

\subsection{Setting}

A set of $n$ qubits corresponds to $\mathbb{C}^{2^n}$ --- a Hamiltonian $H$ is a self-adjoint operator acting on this space, and according to Schrödinger's equation, the state of the system at time $t$, denoted by $\ket{\psi(t)}$, evolves as governed by the differential equation
\begin{equation}
\mathrm{i} \hbar \frac{\partial}{\partial t}\ket{\psi(t)} = H \ket{\psi(t)}.
\end{equation}
Solving this equation for a constant $H$ and initial state $\ket{\psi(0)}$, and throwing real constants into the parameter $t$, we obtain the solution 
\begin{equation} \label{eq:schro}
    \ket{\psi(t)} = e^{\mathrm{i} Ht}\ket{\psi(0)}.
\end{equation}
Certain choices for $H$ allow for a block decomposition in which each block corresponds to the subspace spanned by the global states in which precisely $k$ qubits are at $\ket 1$, and the remaining at $\ket 0$. Thus there are $n+1$ subspaces, one for each $k \in \{0,\cdots,n\}$, each of dimension $\binom{n}{k}$ \cite{Osborne2006}. If $H$ is the $XY$-Hamiltonian (apparently also called $XX$ in some texts), the block corresponding to $k=1$ coincides with the adjacency matrix of the underlying graph that describes the couplings. If $H$ is the Heisenberg Hamiltonian, one observes the Laplacian matrix of the underlying graph in the same $k=1$ block \cite{Christandl2005}.

The $XY$-Hamiltonian is defined in terms of  two-body interactions determined by the edge set of the graph. Given a graph $G = (V,E)$, defined by the vertex set $V$ and edge set $E$, the vertices correspond to qubits, and we can write the Hamiltonian matrix as
\[
	H = \frac{1}{2} \sum_{ab \in E(G)} X_aX_b + Y_aY_b,
\]
where $X_k$ corresponds to the operator which acts as the Pauli matrix $X$ onto the qubit in position $k$, and analogously for $Y_k$ and the Pauli matrix $Y$. It is possible to input real weights multiplying $(X_aX_b + Y_aY_b)$ and still obtain the block decomposition described above, with the Hamiltonian corresponding to the $1$-excitation subspace now being the real symmetric weighted adjacency matrix of a graph. 

\subsection{Encoding direction}

A directed graph can be conveniently represented by means of an Hermitian adjacency matrix. For example, if there is a directed arc from $a$ to $b$, the $(a,b)$ entry can be set to $\ii$ and therefore the $(b,a)$ entry to $-\ii$. In this model it is convenient to replace a pair of arcs of opposing direction between a pair of vertices by an undirected edge of weight $+1$ in the Hermitian adjacency matrix. If instead of $\ii$, one chooses $\e^{\ii \pi/3}$ to encode direction, then the weight for an undirected edge is exactly the sum of the weights of a pair of opposing direction arcs. This convenient fact was explored in \cite{Mohar2017}, in a context unrelated to quantum walks.

Graphs for which orientation of an edge is encoded in the adjacency matrix upon the use of complex numbers of norm 1 are also known as complex unit gain graphs. 

A natural question at this point is whether there is a Hamiltonian model (preferably defined in terms of $2$-body interactions only) whose action onto the $1$-excitation subspace corresponds precisely to a general Hermitian adjacency matrix. The answer is affirmative (see for instance \cite{Zimboras2013}). We add a proof for reference.

\begin{theorem}
	Let $M = ((m_{ab}))$ be an Hermitian matrix, with rows and columns indexed by the vertex set of a graph on $n$ vertices. Then $M$ is a block of the $2^n \times 2^n$ matrix $H$, defined as
	\begin{align*}
	H & = \frac{1}{2} \sum_{a \neq b} \Re(m_{ab}) (X_aX_b + Y_aY_b) + \Im(m_{ab}) (X_aY_b - X_bY_a)  \\
	& + \frac{1}{2} \sum_{a} m_{aa} (I-Z_a) 
	\end{align*}
	Moreover, $M$ corresponds to the action of $H$ onto the subspace spanned by ${\ket a = \ket{0 \cdots 0 1 0 \cdots 0}}$, where the $1$ appears in the $a$th position, for all $a \in V(G)$. 
\end{theorem}
\begin{proof}
	Assume $U \subseteq V(G)$, and by an abuse of notation, let $U \in \{0,1\}^{V(G)}$ also denote the strings of $0$s and $1$s that identifies the vertices in $U$. Fix a subset $S \subseteq V(G)$. Note that 
	\begin{align}
		& \frac{1}{2}(X_aX_b + Y_aY_b) \ket{S} = 
		\left\{
		\begin{array}{ll}
		\ket{S \oplus \{a,b\} } & \text{ if } |S \cap \{a,b\}| = 1, \\ 
		      0 & \text{ otherwise.}
		  \end{array} \right. \\
		& \frac{1}{2}(X_aY_b - X_bY_a) \ket{S} = \left\{
		\begin{array}{ll}
		\ii \ket{S \oplus \{a,b\} } & \text{ if } a \in S, b \notin S, \\ 
		-\ii \ket{S \oplus \{a,b\} } & \text{ if } a \notin S, b \in S, \\
		      0 & \text{ otherwise.}
		  \end{array} \right.\\
		& \frac{1}{2}(I-Z_a) \ket{S} = \left\{
		\begin{array}{ll}
		\ket{S} & \text{ if } a \in S, \\ 
		      0 & \text{ otherwise.}
		  \end{array} \right.
	\end{align}
	It is immediate to verify that $H \ket {S}$ is a linear combination of $\{ \ket{U} : U \subseteq V(G), |U| = |S|\}$, and that $M$ represents the action of $H$ over the subspace determined by the subsets of size $1$.
\end{proof}

The theorem above guarantees that our model is physical, and can be constructed upon the use of $2$-site interactions. From here on, we shall assume the underlying directed graph of the network has been given along with a function $\alpha : E \to [0,2\pi)$ so that the Hermitian adjacency matrix we consider is
\begin{equation}\label{eq:alphaadj}
H_{\alpha} = \sum_{(a,b) \in E(G)} \e^{\ii \alpha(a,b)}\ket{a}\bra{b} + \e^{-\ii \alpha(b,a)}\ket{b}\bra{a}, 
\end{equation}
where we impose the constraint $\alpha(a,b)=\alpha(b,a)$.

\subsection{Constant weight}

Extensive work has been done for when the directed arcs are all encoded with constant complex weight equal to $\ii$ (see for instance \cite{Cameron2014,Connelly2017,godsil2020perfect}). We show at least one more general case that can be reduced to this. 

By considering a fixed $\alpha$ to every edge defined in \eqref{eq:alphaadj}, we can split the Hamiltonian into two parts as
\begin{equation}
H_{\alpha} = \e^{\ii \alpha} B + \e^{-\ii \alpha} B^T,
\end{equation}
where $B^{V(G) \times V(G)}$ consists of a $01$ matrix with $1$ in position $(a,b)$ for every arc $(a,b) \in E(G)$, and $0$ elsewhere, including in position $(b,a)$. We may expand $\e^{\ii\alpha}$ to find
\begin{equation}
H_{\alpha} = \cos(\alpha)[B+B^T] + \sin(\alpha)[\ii \,(B-B^T)].
\end{equation}
Note in particular that $[B+B^T]$ is the adjacency matrix of the underlying undirected graph, whereas $[\ii \,(B-B^T)]$ is precisely equal to the Hermitian adjacency matrix for which constant weight $\ii$ has been chosen to encode direction. Even though the Hamiltonian is so nicely decomposed, the transition matrix of the quantum walk might not, as generally $B$ and $B^T$ do not commute. If they do, however, which corresponds to the case where $B$ is a normal matrix, the transition matrix of the quantum walk, as in \eqref{eq:schro}, will be given as
\[
	\e^{\ii t H_{\alpha}} = \e^{\ii t \cos(\alpha)[B+B^T]}\e^{\ii t \sin(\alpha)[\ii \,(B-B^T)]},
\]
where each factor in the product is the transition matrix of a quantum walk on a suitable (Hermitian) Hamiltonian. 

Standard examples of normal matrices $B$ that will provide interesting cases are the sums of circulant matrices (which can all be simultaneously diagonalizable with Fourier coefficients). For instance, if $G$ is a complete graph (all vertices are connected) on $n$ vertices, with $n$ odd, it is possible to write
\[
	A(G) = \sum_i B_i + B_i^T,
\]
where each $B_i$ is a circulant matrix. Therefore certain orientations of the edges of the complete graph given by weights $\e^{\ii \alpha}$ can be cast into the above framework.

In the next subsection, we display an interesting phenomenon for when orientation of the edges of the complete graph cannot be interpreted as such, suggesting that freedom in choosing the weights $\alpha$ might lead to interesting cases.

\subsection{A first example: quantum walk on a complete graph}

We show how the presence of weights changes the transition probability in complete graphs. Let $K_n$ denote the complete graph on vertex set $\{1,...,n\}$. We use $\ket i$ to denote the indicator function of vertex $i$. It corresponds to the global state that assigns some state to the qubit at $i$, and orthogonal states to the remaining qubits. Assume all edges of the graph have been oriented from $i \to j$ if $i < j$, and this weight is given by $\e^{\ii \alpha}$. Thus the Hamiltonian is
\[
	H = H_\alpha = \sum_{i < j} \e^{\ii \alpha}\ket{i}\bra{j} + \e^{-\ii \alpha}\ket{j}\bra{i} =  \e^{\ii \alpha} B + \e^{-\ii \alpha} B^T, 
\]
where $B$ is the all ones upper triangular matrix (with $0$ diagonal). The graph in Figure \ref{fig:k4walk} depicts the value of
\[
	P_{0\to 1}(t) = |\bra{1} \e^{\ii t H} \ket{0}|^2,
\]
for increasing $t$, considering the cliques $K_4$ and $K_6$, and different choices of $\alpha$, where $2\pi$ is the non-oriented model.

\begin{figure}[!h]
\centering
\includegraphics[width=10cm]{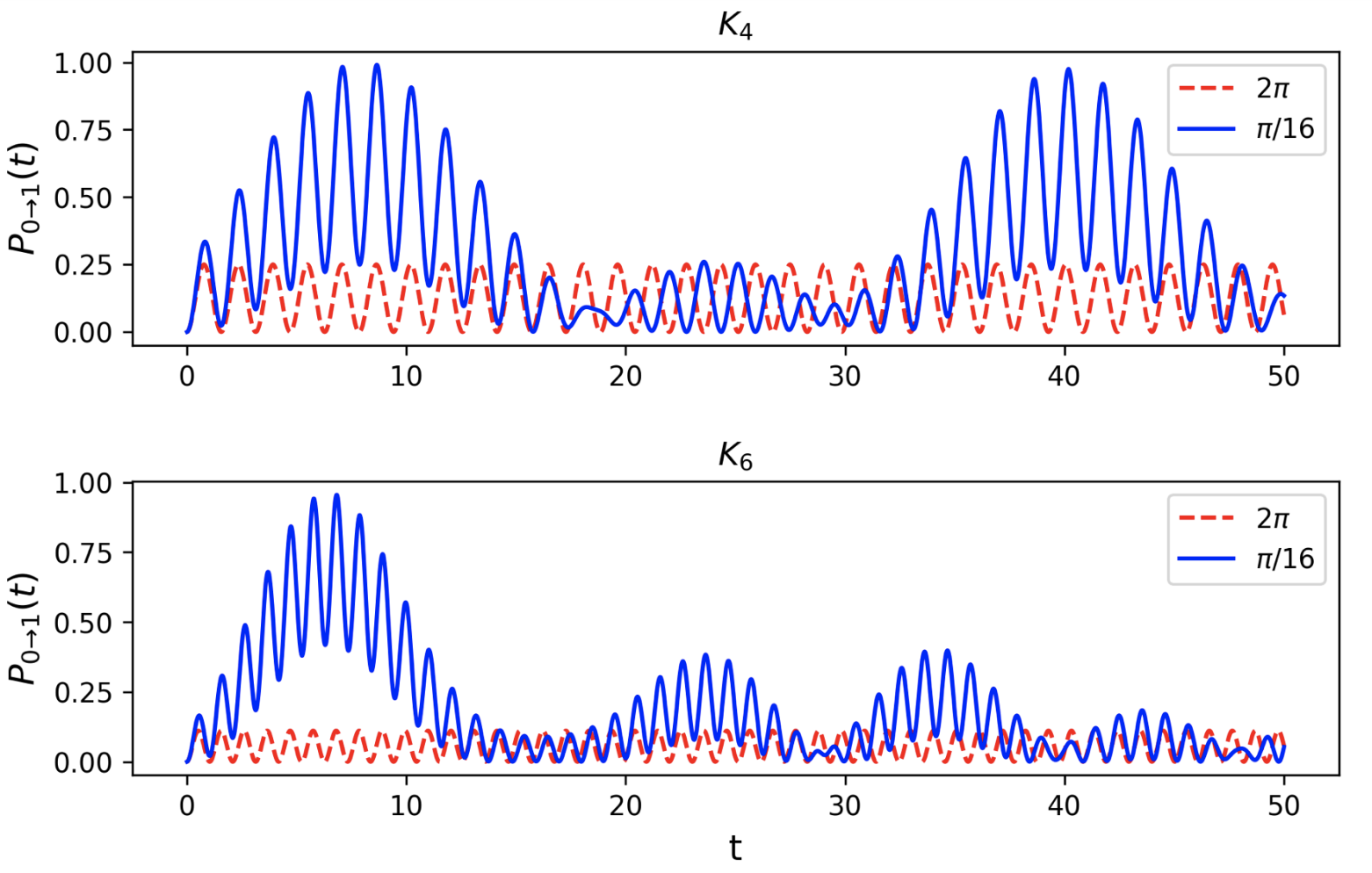}
\caption{Dynamics of a continuous-time quantum walk on $K_6$ with different values of $\alpha$.}
\label{fig:k4walk}
\end{figure}

This suggests that orientation can be a useful tool in increasing the probability of state transfer.

\section{Orientations on trees}

The typical example when studying quantum walks is for the Hamiltonian to be defined based on a linear chain of nearest neighbour interacting qubits. More generally, for the purposes of this section, we assume the underlying graph is a tree (an acyclic connected graph). We show below that adding arbitrary orientations given by complex numbers of absolute value $1$ affects the transition matrix of the quantum walk in a very predictable manner (in explicit contrast to what occurs if weights put to edges have absolute value different than $1$).

The following lemma is elementary, but its conclusion is relevant to our context.

\begin{lemma}\label{lemmatree}
    Let $H_{\alpha}$ be the adjacency matrix of an oriented tree $T$ on $n$ vertices, where each arc $(a,b)$ has received weight $\e^{\ii \alpha(a,b)}$ (recall that the $(b,a)$ of $H_{\alpha}$ is thus equal to $\e^{-\ii \alpha(a,b)}$). Then, there is a diagonal matrix $D$, that obeys $D^\dagger D=I$, so that
    \begin{equation}
        D^\dagger H_{\alpha}D = H_0,
    \end{equation}
    where $H_0$ is the adjacency matrix of an undirected underlying tree.
\end{lemma}
\begin{proof}
	The proof goes by induction on the number of vertices. The base case for a tree on $2$ vertices is trivial. Assume $a$ is a leaf of $T$, connected to $b$ (say by an arc $(a,b)$). Let $E$ be the $(n-1)\times (n-1)$ diagonal matrix that gives $E^\dagger H_{\alpha}(T -v) E = H_0(T-v)$. Let $D$ be obtained from $E$ upon appending one diagonal entry corresponding to vertex $a$ of $T$, so that 
	\[
		D_{aa} = (\overline{E_{bb}} \cdot \e^{\ii \alpha(a,b)})^{-1}.
	\] 
	It is immediate to verify that $|D_{aa}| = 1$, and that $D^\dagger H_{\alpha}(T) D = H_0(T)$.
\end{proof}

As a consequence, we have that
\[
	\e^{\ii t H_\alpha} = D \, \e^{\ii t H_0}\, D^\dagger.
\]
This shows that if the initial state of the walk is of the form $\ket{a}$ for some vertex $a$ of the tree, then adding orientations does not affect the probabilities that this state is observed elsewhere in the tree after a given time.

Kubota \textit{et al.}\cite{Kubota2021} essentially showed that when $\alpha$ is constant, the matrices are similar (though they did not explicitly used diagonal similarity, which leads to equivalence of walks). The proof we present above is different from theirs, but we believe the method they use to analyze graphs with a given size of shortest cycle should also work for when $\alpha$ is possibly non-constant.

\section{Quantum walk on a cycle}

\subsection{Setting}
Let us consider the case where $\alpha$ assumes the same value to each arc and the conjugate to another. The hamiltonian of cycle, based on its adjacency matrix, is defined by
\begin{eqnarray}
H_{\alpha} = \sum_{x = 0}^{N-1} & &e^{i\alpha(x+1,x)}\ket{x+1}\bra{x} + e^{-i\alpha(x,x+1)}\ket{x}\bra{x+1}
\end{eqnarray}
as we can see its representation in figure \ref{fig:oriented_cycle}, where we consider a cyclic boundary condition by performing addition modulo $N$, and $\alpha(u,v) = \alpha(v,u)$
\begin{figure}[!h]
\centering
\includegraphics[width=6cm]{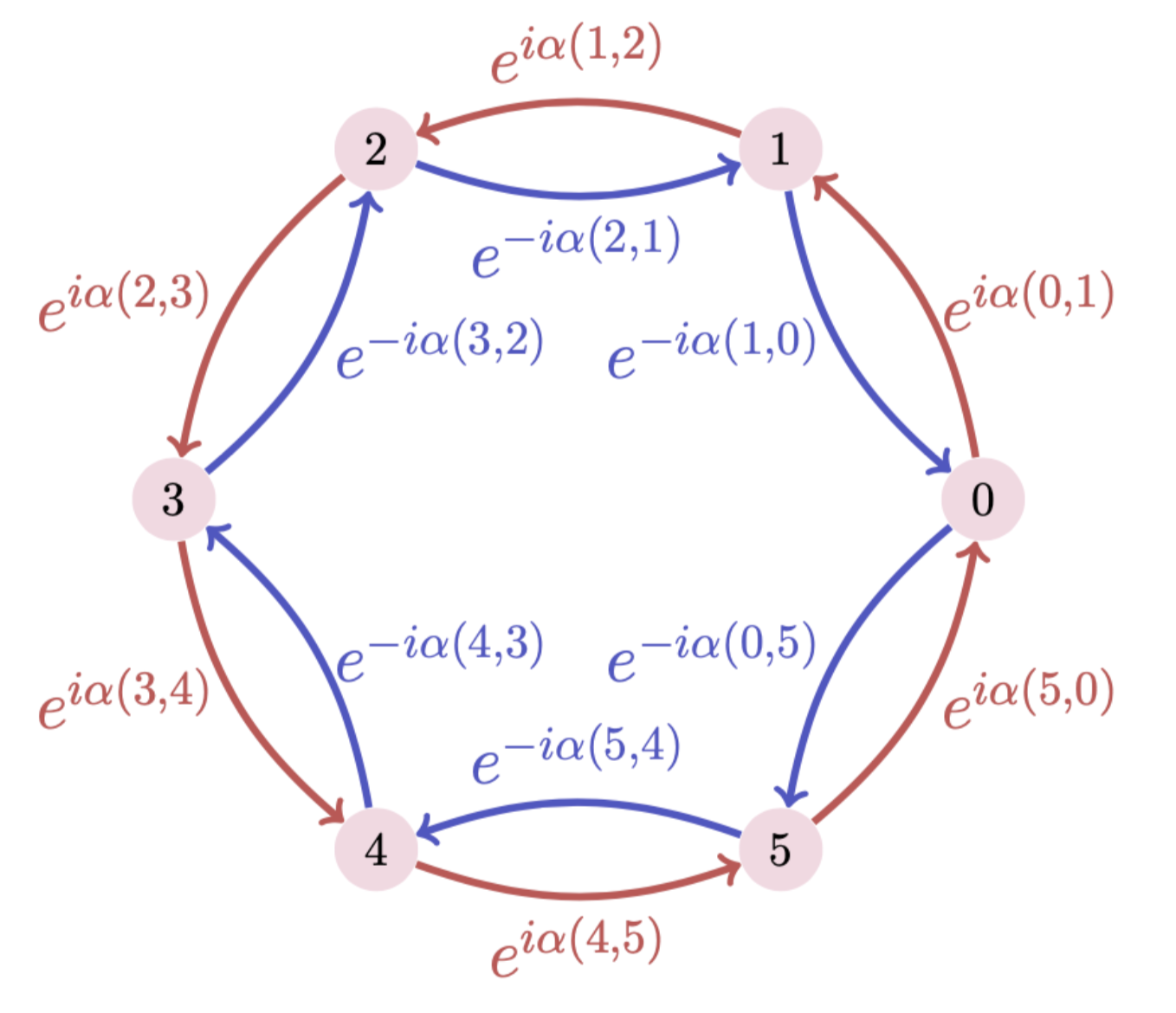}
\caption{Oriented cycle graph with six vertices}
\label{fig:oriented_cycle}
\end{figure}

This model of Quantum Walk can produce new interference patterns and it can helps to find new transport properties. In order to study this behaviour, we consider the initial condition
\begin{equation}\label{eq:initialab}
\ket{\psi(0)} = \ket{0}.
\end{equation}
Figure \ref{fig:dyn_inf_line1} depicts the behaviour of the quantum walk considering over a cycle graph with 26 vertices and the same $\alpha$ for all of the edges, and the initial condition in \eqref{eq:initialab}. This figure illustrates how this parameter $\alpha$ affects the dynamics and, depending on its value, gives a more centralized behaviour or certain tendency to right or left side over the line.
\begin{figure}[!h]
\centering
\includegraphics[trim={1.4cm 0.2cm 4cm 1cm},clip,height=5.0cm]{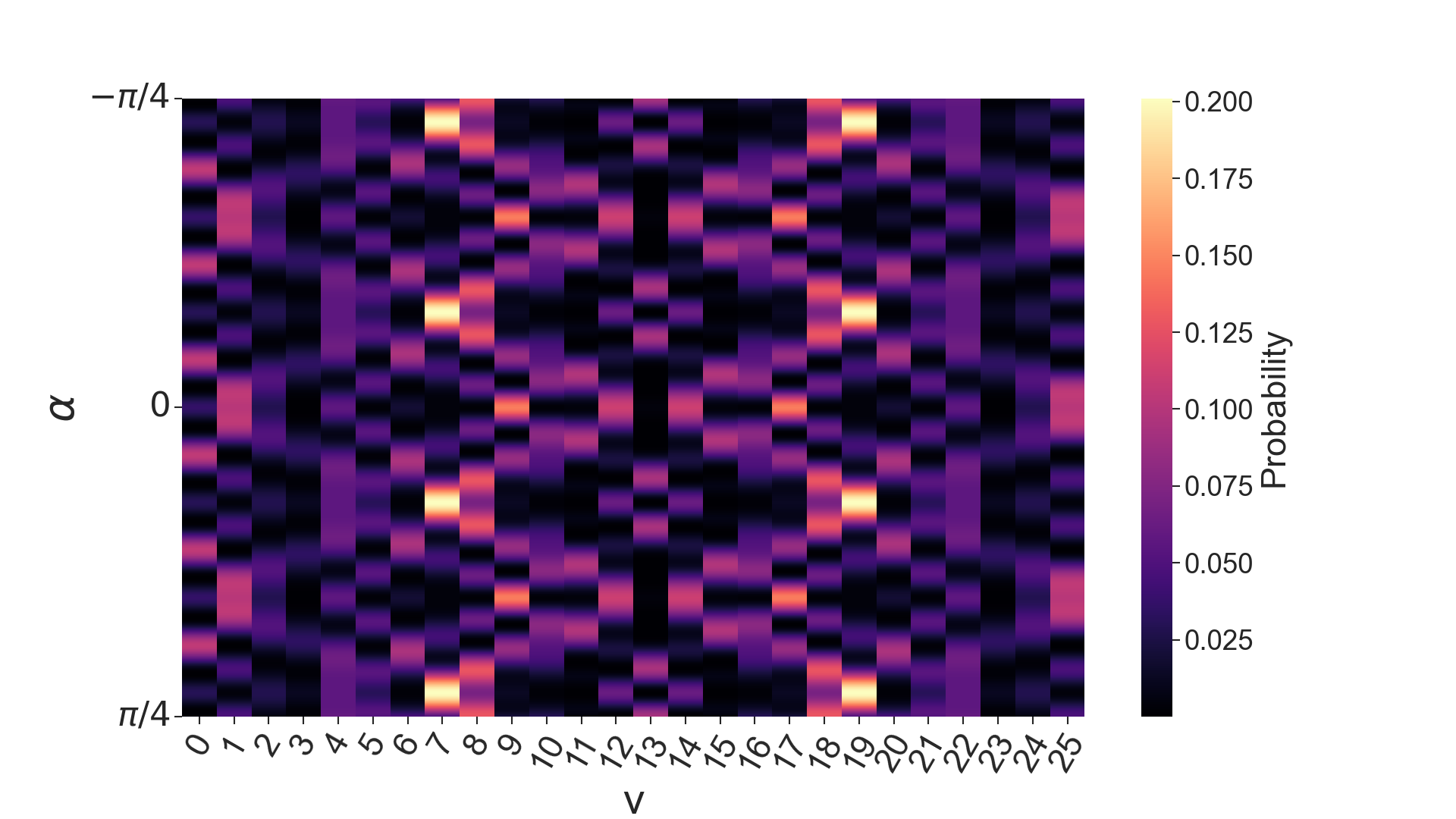}
\caption{Dynamics of a continuous-time quantum walk on a cycle graph with different values of $\alpha$, time equals to $10\pi$ and initial condition $\ket{0}$.}
\label{fig:dyn_inf_line1}
\end{figure}

\subsection{\label{sec:ztrans}Zero Transfer in even cycles}

Given a graph, extensive work has been done in studying when an input state at a given qubit can be transferred to another with maximum (perfect state transfer, see \cite{kay2010perfect,kendon2011perfect} for some surveys) or almost maximum probability (pretty good state transfer, see \cite{kempton2017pretty,banchi2017pretty} for some recent work). The somewhat symmetrical problem of asking when the probability of transfer between two vertices is $0$ seems to have received less attention. See for instance \cite{Zimboras2013}. If this probability is required to be constant equal to $0$ for all times, the phenomenon has been called zero transfer, and was studied in \cite{Sett2019}.

Given a Hamiltonian $H_\alpha$ for a finite graph, as we defined in Section \ref{sec:ctqw2}, we say that zero transfer occurs between vertices $a$ and $b$ if
\begin{equation}
    \abs{\bra{b}\e^{\ii t H_\alpha}\ket{a}}=0,
\end{equation}
for all $t$. The aim of this section is to prove that there is zero transfer between antipodal vertices in cycle graphs with $2m$ vertices if the product of the weights following an orientation of the cycle is equal to $-1$. This is typically achieved, for example, if $k$ of its arcs have been oriented in the same direction with weight $\e^{\ii \alpha}$ where $\alpha = \pi/k$, and the other remaining vertices have been left unaltered. This result strongly generalizes the known examples of zero transfer for when one of the edges is signed with $-1$ (see \cite{Sett2019}).

\begin{figure}[!h]
\centering
\includegraphics[trim={1.4cm 0.2cm 4cm 1cm},clip,height=5.0cm]{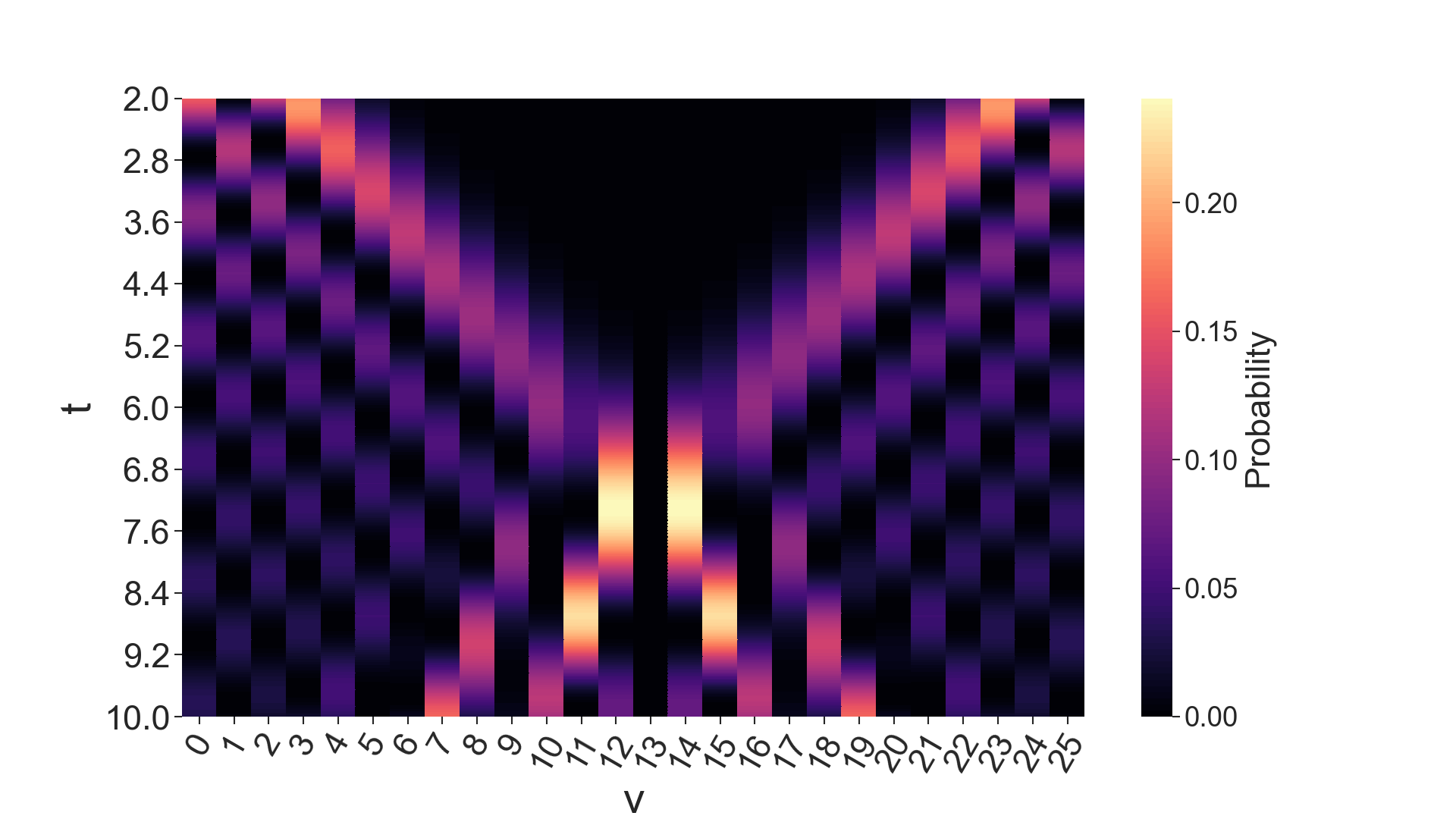}
\caption{Dynamics of a continuous-time quantum walk on a cycle graph with initial condition equals to $\ket{0}$ and $\alpha=\pi/k$.}
\label{fig:dyn_inf_line}
\end{figure}


The following widely known application of the Laplace expansion for the determinant will be  useful to us.
\begin{lemma}\label{lemmatridig}
    Let $T_n$ be a tridiagonal matrix given by
    \begin{equation}
        T_n=\begin{pmatrix}
            a_1&b_1&&&\\
			c_1&a_2&b_2&&\\
			&c_2&\ddots&\ddots&\\
			&&\ddots&\ddots&b_{n-1}\\
			&&&c_{n-1}&a_n
        \end{pmatrix},
    \end{equation}
    and define $T_k$ to be the principal $k \times k$ block of $T_n$ containing its first $k$ rows and columns. Then, denoting $\det(T_k)=f_k$, we have
    \begin{equation}
        f_k=a_kf_{k-1}-c_{k-1}b_{k-1}f_{k-2},
    \end{equation}
    where $f_{-1}=0$ and $f_0=1$.
\end{lemma}
A direct consequence, that will be useful later on, occurs when $a_k=0$ and $b_k=c_k=1$ for all $k$, which leads to
\begin{gather}
    f_{2k}=(-1)^k, \label{eq:detpaths1}\\
    f_{2k+1}=0.\label{eq:detpaths2}
\end{gather}

Before moving on, we recall the definition of the Chebyshev polynomial of first kind, $T_n(x)$, and second kind, $U_n(x)$, as
\begin{definition}
The Chebyshev polynomial of first kind is defined as
\begin{equation}
    T_n(x)=\sum_{k=0}^{\lfloor\frac{n}{2}\rfloor}\binom{n}{2k}(1-x^{-2})^k,
\end{equation}
while the Chebyshev polynomials of second kind is
\begin{equation}
    U_n(x)=\sum_{k=0}^{\lfloor\frac{n}{2}\rfloor}(-1)^k \binom{n-k}{k}(2k)^{n-2k}.
\end{equation}
\end{definition}
Those two polynomials are related by the following known properties (see \cite{Mason2002}):
\begin{gather}\label{chybprop}
    T_n(x)=\frac{1}{2}\bigg(U_n (x)-U_{n-2}(x)\bigg),\\
    T_{2n}(x)=2T^2_n (x)-1.
\end{gather}

We are now ready to prove the following
\begin{theorem}\label{thm:detcycle}
    Let $H_\alpha$ be the adjacency matrix of a cycle with $n = 2m$ vertices, with vertices $v_1,...,v_{n}$, where $v_j$ is adjacent to $v_{j \pm 1}$. Suppose
    \[
    	\sum_{j = 1}^n \alpha(j,j+1) = \pi.
    \]
    Then, the characteristic polynomial $\phi_{H_\alpha}$ is given by
    \begin{equation}
        \phi_{H_\alpha}=\phi_H+4,
    \end{equation}
    where $\phi_H$ is the characteristic polynomial of an undirected cycle with $2m$ vertices. Additionally, $\phi_{H_\alpha}$ can be decomposed as
    \begin{equation}
        \phi_{H_\alpha}=(2T_k(x/2))^2.
    \end{equation}
\end{theorem}
\begin{proof}
    The Leibniz formula for the determinant gives that
    \begin{equation}
        \phi_{H_\alpha}=\sum_{\sigma\in\mathcal{S}}(-1)^{\textbf{sgn}(\sigma)}\prod_{i=1}^{2k}(xI-H_\alpha )_{i\sigma(i)},
    \end{equation}
    where the sum corresponds to all permutations $\sigma$ of the set $\{1,2,\dots,2m\}$. Notice that the permutation will only give a non-zero entry of $(xI-H_\alpha)$ if it maps $i$ to a neighbour. Therefore, we can rewrite
    \begin{equation}
        \phi_{H_\alpha} = \sum_{D \subseteq V(G)} x^{\abs{D}}(-1)^{n-\abs{D}}\det(H_{\alpha}\setminus D),
    \end{equation}
    where the sum runs over all subsets $D$ of the vertices of the graph, and $(H_{\alpha}\setminus D)$ denotes the matrix $H_{\alpha}$ with rows and columns corresponding to $D$ removed.
    
    It is easy to see that subgraphs obtained for $\abs{D}< n$ are all paths or multiple disconnected paths where some edges might have been weighted. As we saw on Lemma \ref{lemmatree}, all these possibly weighted directed paths are similar (via a diagonal matrix) to their undirected unweighted counterparts. Therefore
    \begin{equation}
        \phi_{H_\alpha}=\phi_{H}-\det(H)+\det(H_{\alpha}).
    \end{equation}
    
    To compute the determinant of $H_{\alpha}$, it suffices to observe the following decomposition, which follows the Laplace expansion (see for instance \cite{Godsil1993}):
    
    \begin{align}
    	\det(H_\alpha)= & -(H_\alpha)_{12}(H_\alpha)_{21} \det((H_\alpha) \setminus \{1,2\})  \nonumber \\ 
    	& + (-1)^{n+1}\cdot 2\cdot  \prod_{j = 1}^n (H_\alpha)_{j\, j+1} \nonumber \\ 
    	& -(H_\alpha)_{1n}(H_\alpha)_{n1} \det((H_\alpha) \setminus \{1,n\}) \label{eq:H}
    \end{align}
    
    From the given weights, Lemma \ref{lemmatree} and from \eqref{eq:detpaths1}, we have
    
    \begin{itemize}
		\item $(H_\alpha)_{12}(H_\alpha)_{21} = (H_\alpha)_{1n}(H_\alpha)_{n1} = 1$.
		\item $\det((H_\alpha) \setminus \{1,2\}) = \det((H_\alpha) \setminus \{1,n\}) = (-1)^m$.
		\item $\prod_{j} (H_\alpha)_{j\, j+1} = -1$.
    \end{itemize}
    
    If instead we were computing the determinant of $H$, the only difference would have been that $\prod_{j} (H)_{j\, j+1} = 1$. Therefore
    \begin{equation}
        \phi_{H_{\alpha}}=\phi_{H}+4.
    \end{equation}
    
    Also following from \eqref{eq:H}, we have
    \begin{equation}
        \phi_H=\phi_{P_{n}}-\phi_{P_{n-2}}-2,
    \end{equation}
    where $\phi_{P_{n}}$ is the characteristic polynomial of the adjacency matrix of the path graph with $n = 2m$ vertices, hence, from \eqref{eq:detpaths1}, we have
    \begin{equation}
        \det(H)=\begin{cases}
        0&\text{if $m$ is even},\\
        -4&\text{if $m$ is odd}
        \end{cases}.
    \end{equation}
    Also, $\phi_{P_{n}}$ can be defined in terms of the Chebyshev polynomials as
    \begin{equation}
        \phi_{P_{n}}=U_{n}(x/2).
    \end{equation}
    Using that in the expression for $\phi_{H_\alpha}$ and the properties presented in the definition of the Chebyshev polynomials we have
    \begin{equation}
        \phi_{H_\alpha}=(2T_k(x/2))^2.
    \end{equation}
\end{proof}


A direct consequence of this theorem is
\begin{corollary}
Let $H_\alpha$ be the adjacency matrix of a cycle with $n = 2m$ vertices where $\alpha$ satisfies the hypothesis in Theorem \ref{thm:detcycle}. Then there is zero transfer for any time $t$ between any vertex $a \in \{1,...,n\}$ and its antipodal vertex $(a+m) \pmod n$.
\end{corollary}
\begin{proof}
    Take a vertex $a$, and consider its indicator vector $\ket a$. The minimal polynomial of $H_\alpha$ on $\ket a$ is the monic polynomial $p(x)$ of smallest degree so that $p(H_\alpha) \ket a = 0$. It is well known that $p(x)$ divides the minimal polynomial of the matrix $H_\alpha$, and this latter polynomial has degree equal to the number of distinct eigenvalues of $H_\alpha$. From Theorem \ref{thm:detcycle}, this number is at most $m$. Therefore, $H_{\alpha}^\ell \ket a$, for $\ell \geq m$, is a linear combination of $H_{\alpha}^k \ket a$ for $0 \leq k \leq m-1$. Because the combinatorial distance between $a$ and $(a+m)$ is $m$, it follows that $\bra {a+m} H_{\alpha}^k \ket a = 0 $ for all $0 \leq k \leq m-1$, and therefore $\bra {a+m} H_{\alpha}^\ell \ket a = 0$ for all $\ell \geq 0$. This immediately implies that for all idempotents $E_r$ in the spectral decomposition of $H_\alpha$, this corresponding entry related to $a$ and $(a+m)$ is $0$, and therefore
    \begin{equation}
        \abs{\bra{a+m}\e^{\ii t H_{\alpha}}\ket{a}}^2=\abs{\sum_{r}\e^{\ii t \theta_r}\bra{a+m}E_r\ket{a}}^2=0.
    \end{equation}
\end{proof}


\begin{figure}[!b]
\centering
\includegraphics[width=10cm]{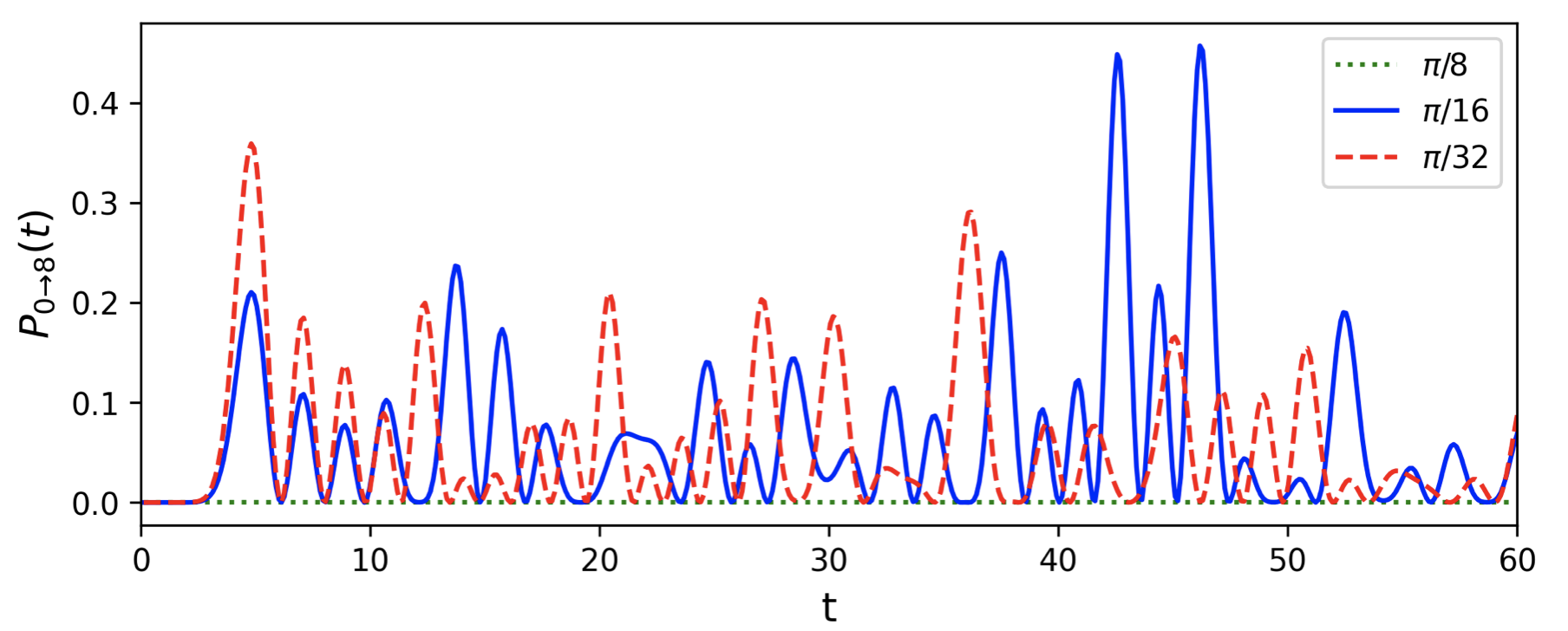}
\caption{$C_{16}$ with weights $\alpha=\pi/8$, $\alpha=\pi/16$, $\alpha=\pi/32$.}
\label{fig:ztrans}
\end{figure}

In Figure \ref{fig:ztrans} we display the probability of transfer between antipodal vertices in $C_{10}$ when half of their arcs have received the indicated weights, while the other half is maintained with weight $1$. To show our results, the values of $\alpha$ are changed and it is possible to see that zero transfer disappears if $\alpha\neq\pi/m$.

Kubota \textit{et al.}\cite{Kubota2021} computed the characteristic polynomial of cycles that have received constant weights on some of their arcs, but they did not showed this connection to zero transfer.

\section{Conclusion}

We have shown how to obtain zero transfer in more general directed cycles than previously observed in the literature. As zero transfer is an interesting and desirable phenomenon, we believe that a desirable research target is a full characterization of graphs that admit it. Our result gives a possible path for its complete characterization and the connection with the spectral properties of the adjacency matrix of the graph.

We have also displayed interesting examples of directed complete graphs whose transfer probability between two vertices approaches one, as opposed to their undirected counterparts. This suggests that further investigation on the possible features directed arcs bring to quantum walks is imperative.

Finally, we displayed the viability of the proposed formalization as a physical realizable Hamiltonian with only two-body interactions. 

\section{Acknowledgements}

This study was financed by the Coordenação de Aperfeiçoamento de Pessoal de Nível Superior – Brasil (CAPES) – Finance Code 001

\bibliographystyle{acm}
\bibliography{bibliography}

\end{document}